\newif\ifabstract
\abstracttrue
 \abstractfalse 
\newif\iffull
\ifabstract \fullfalse \else \fulltrue \fi

\documentclass[11pt]{article}

\advance \oddsidemargin by -0.8in

\usepackage{microtype}
\usepackage{graphicx}
\usepackage{subfigure}
\usepackage{booktabs} 

\usepackage[numbers]{natbib}

\usepackage{hyperref}

\usepackage{amsmath}
\usepackage{amssymb}
\usepackage{mathtools}
\usepackage{amsthm}

\usepackage[capitalize,noabbrev]{cleveref}

\theoremstyle{plain}
\newtheorem{theorem}{Theorem}[section]
\newtheorem{proposition}[theorem]{Proposition}
\newtheorem{lemma}[theorem]{Lemma}

\theoremstyle{definition}
\newtheorem{definition}[theorem]{Definition}

\theoremstyle{remark}

\usepackage[textsize=tiny]{todonotes}

\usepackage{enumitem}
\usepackage[utf8]{inputenc} 
\usepackage[T1]{fontenc}    
\usepackage{hyperref}       
\usepackage{url}            
\usepackage{booktabs}       
\usepackage{amsfonts}       
\usepackage{nicefrac}       
\usepackage{microtype}      
\usepackage{xcolor}         

\usepackage{tikz}
\usepackage{amsthm}
\usepackage{amsmath}

\usepackage{xspace}

\usepackage{graphicx}
\usepackage{subfigure}

\usepackage{bm}

\newcommand{\DSName}{\texttt{RobustSL}\xspace}

\newcommand{\Exp}{\mathbb{E}}
\newcommand{\BaseClass}[1]{\texttt{class-base(#1)}\xspace}
\newcommand{\pz}{\theta}
\newcommand{\pg}{p}

\newcommand{\Oh}{\mathcal{O}}

\definecolor{Green}{rgb}{0.0, 0.4, 0.0}

\begin{document}
\title{Robust Learning-Augmented Dictionaries}

\author{
Ali~Zeynali\thanks{University of Massachusetts Amherst. Email: {\tt azeynali@cs.umass.edu}.} \and
Shahin~Kamali\thanks{York University. Email: {\tt kamalis@yorku.ca}.}  \and
Mohammad~Hajiesmaili\thanks{University of Massachusetts Amherst. Email: {\tt hajiesmaili@cs.umass.edu}.} \and
}

\begin{titlepage}
\maketitle

\begin{abstract}
We present the first learning-augmented data structure for implementing dictionaries with optimal consistency and robustness. Our data structure, named \DSName,
is a skip list augmented by predictions of access frequencies of elements in a data sequence. With proper predictions, \DSName has optimal consistency (achieves {static optimality}).
At the same time, it maintains a logarithmic running time for each operation, ensuring optimal robustness, even if predictions are generated adversarially. Therefore, \DSName has all the advantages of the recent learning-augmented data structures of
Lin, Luo, and Woodruff (ICML 2022) and Cao et al. (arXiv 2023), while providing robustness guarantees that are absent in the previous work. Numerical experiments show that \DSName outperforms alternative data structures using both synthetic and real datasets.
\end{abstract}

\end{titlepage}

\section{Introduction}
 Dictionaries are one of the most studied abstract data types with a wide range of applications, from database indexing~\cite{BayerM72} to scheduling in operating systems~\cite{Rubini}. In the dictionary problem, the goal is to maintain a set of $n$ \emph{items}, each represented with a \emph{key}, so as to minimize the total \emph{cost} of processing an online data sequence of operations, each involving access, insertion, deletion, or other operations such as successor and range queries. Here, the cost refers to the number of comparisons made for all operations, and \emph{amortized cost} is the average cost of a single operation over the input sequence. 

{There are several data structures to implement dictionaries.} 
Hash tables are {efficient} and practical data structures that are useful in many applications. {However, hash tables 
do not efficiently support operations like successor and rank/select queries}~\cite{LinLWood22}. In comparison, binary search trees (BST) and skip lists keep items sorted and thus support queries that involve ordering items with minimal augmentation. Classic BSTs such as red-black trees, AVL trees, treaps, and classic skip lists support dictionary operations in $\Oh(\log n)$ time, which is optimal for a \emph{single} operation. For a \emph{sequence} of $m$ operations in a data stream, however, these structures are sub-optimal as they do not react to the access patterns in the input. 
For example, an input may be formed by $m$ requests to a  leaf $i$ of a balanced BST, giving it a total cost of $\Theta(m \log n)$, while 
a solution that first moves $i$ to the root has a 
cost of $\Theta(m + \log n)$. 
The same argument can be made by making 
requests to a key replicated once at the deepest level of a skip list.
Splay trees~\cite{SleatorT85} are self-adjusting BSTs that move each accessed item to the tree's root via a
splay operation. 
Splay trees are \emph{statically optimal}, meaning their cost is proportional to an optimal data structure that does not self-adjust.
Nevertheless, as pointed out by~\citet{LinLWood22} and~\citet{ChenChen2023V2}, the constant multiplicative overhead involved in pointer updates in the splay operations makes them impractical in many applications. 

In recent years, there has been an increasing interest in augmenting algorithms that work on data streams with machine-learned predictions. 
The objective is to design solutions that provide guarantees with respect to \emph{consistency}, the performance measure when predictions are accurate, and \emph{robustness}, the performance measure when predictions are adversarial. 
For dictionaries, ~\citet{LinLWood22} presented an augmented treap data structure that uses frequency predictions. In a classic treap, each item is assigned a random priority that defines its location in the tree. In the work of~\citet{LinLWood22}, these random priorities are replaced with machine-learned frequency predictions. 
Under the ``random-order rank'' assumption, which implies all keys have the same chance of being the $i$'th frequently asked key (for any $i\leq n$), the resulting data structure 
offers static optimality with accurate predictions (is optimally consistent) and is
robust in the sense that the expected cost of all operations in $\Oh(m \log n)$. 
When the ranks are not random, however, the resulting data structure is not consistent nor robust (\S\ref{sect:learningTreaps}, Proposition~\ref{prop:ICML}). 
~\citet{ChenChen2023V2} have recently proposed an alternative priority assignment that results in a learning-augmented treap that is 
statically optimal with accurate predictions without the random-order rank assumption. Nevertheless, we show that their structure is also vulnerable when predictions are adversarial and thus is not robust (\S\ref{sect:learningTreaps}, Proposition~\ref{prop:chenchen}).

\paragraph{Contributions.}
This work introduces \DSName, a learning-augmented data structure that utilizes frequency predictions and achieves optimal consistency and robustness. Our primary contributions are delineated as follows:
\begin{itemize}[leftmargin=5mm]

\item In \S~\ref{sec:robustSL}, we presented \DSName, a skip list that replicates each item, given its predicted frequency, based on two factors. First, a deterministic factor is calculated based on the predicted frequency of the item using a classification approach that ensures items with similar predicted frequencies receive a comparable number of replicas within the skip list. In particular, including this factor guarantees that items with higher predicted frequencies obtain more expected number of replicas, thereby optimizing the structure's performance. The second factor is a randomized factor added to the deterministic one. The intuition behind this factor is similar to classic skip lists: adding a random ``noise'' to the number of replicas to ensure efficient list navigation when searching for an item. The deterministic replication process ensures consistency, whereas the stochastic mechanism fortifies the robustness of \DSName.

\item 
We then analyze the consistency and robustness of \DSName. Theorem~\ref{thm:rsl_consistency} establishes consistency of the \DSName by revealing that under perfect frequency predictions, the anticipated access cost for an item $i$ within \DSName is proportional to the logarithm of its predicted frequency, which is a consequence of the deterministic replication mechanism. This relationship confirms that the access cost for an input sequence is a constant factor away from the entropy of the input, thus proving the optimal consistency of \DSName. Additionally, our analysis demonstrates that the maximum cost of access to any item in \DSName remains within $\Oh(\log n)$, substantiated by the stochastic replication mechanism), solidifying the structure's assured optimal robustness (Theorem~\ref{thm:rsl_robustness}). 

\item In~\S\ref{sect:exp}, we conduct comprehensive experimental comparisons between \DSName and other dictionary data structures using synthetic and real datasets. Our experimental findings align closely with theoretical analysis, demonstrating that data structures such as learned treaps exhibit satisfactory performance
for inputs with highly accurate predictability. However, their performance is significantly worse when predictions are erroneous.
In comparison, \DSName achieves comparable performance with accurate predictions and remains robust even for highly erroneous predictions.

\end{itemize}

\section{Preliminaries}
We use $[n]$ to denote the set $\{1,2,\ldots,n\}$ and assume keys in the dictionary come from universe $[n]$. 
We use $m$ to denote the number of operations in the input stream. 
We also let $f_i$ denote the frequency of operations involving key ${i\in[n]}$ in the input, and ${\hat{ f_i}}$ denote the predicted frequency for queries to $i$ ($\sum_i f_i = \sum_i \hat{f_i} = 1$). We let $\bm{f}$, and $\bm{\hat{f}}$ denote the vectors of actual and predicted frequencies. 
Let $e_i \in [n]$ be the item with {\em rank} $i$; that is, the $i$'th most frequently-accessed item in the input.
Finally, we use $\texttt{Operation}_D(i)$ to denote the number of comparisons for applying an operation (the ``\texttt{Operation}'') involving key $i\in[n]$ in a data structure $D$. For example, $\texttt{Search}_D(i)$ is the number of comparisons for accessing $i$ in $D$.
\subsection{Dictionaries and Optimality}
Consider an online stream of operations concerning a set of $n$ items, forming a dictionary. The operations mainly 
%
%
require accessing, inserting, and deleting items, while secondary operations such as rank, select, range-query, successor, and predecessor may need to be supported. 
\emph{Comparison-based} data structures, such as BSTs and skip lists, keep data in order and thus support all operations without relying on restrictive assumptions on the input distribution.
The cost of these data structures for a given input stream can be measured by the number of comparisons they make for all operations in the input; other costs, such as pointer updates, are proportional to the number of comparisons. 
Most comparison-based data structures, such as balanced BSTs and skip lists, 
do not change structure after access operations. \emph{Self-adjusting} data structures such as splay trees~\cite{SleatorT85}, on the other hand, 
require extra comparisons for self-adjusting and thus have a constant-factor overhead, which is not desirable in many applications~\cite{LinLWood22}.

\citet{Mehlhorn75} showed that the number of comparisons made by any comparison-based data structure {that does not self-adjust}, for the input of $m$ operations and frequency distribution $\bm{f}$, is at least $m H(\bm{f})/\sqrt{3}$, where $H(\bm{f})$ is the {\em entropy} of $\bm{f}$  defined as $H(\bm{f}) = -\sum_{i \in[n]} f_i\log(f_i).$  
Therefore, one can use entropy as a reference point for measuring the consistency of learning-augmented structures. In particular, we say that a structure is {\em statically optimal} if its cost for {any} instance $\mathcal{I}$ is $\Oh(H(\bm{f}))$. 
In particular, the weight-balanced BST of~\citet{Knuth71} and its approximation by \citet{Mehlhorn75} are statically optimal.

To measure {robustness}, i.e., performance under adversarial error, 
we note that, regardless of the structure comparison-based structure, there are access requests that require at least $\lceil \log n \rceil$ per access due to the inherent nature of binary comparisons. Therefore, when predictions are adversarial, there are worst-case input sequences in which every access request takes at least $\lceil \log n\rceil$ comparisons ($\Omega(m \log n)$ cost for a sequence of $m$ operations). 
From the above discussions, we define our measures of optimality as follows.

\begin{definition}[Optimal consistency and robustness]
An instance $\mathcal{I} = (\sigma,\bm{f}, \hat{\bm{f}})$ of the dictionary problem with prediction includes a sequence $\sigma$ of $m$ operations on $n$ keys, an unknown vector $\bm{f}$ specifying frequency (probability) of keys in $\sigma$, and a known vector $\hat{\bm{f}}$ of predicted frequencies. 
A learning-augmented data structure has consistency $c$ iff its total number of comparison is $c \cdot H(\bm{f})$ 
for any input instance with $\bm{f} = \hat{\bm{f}}$. In particular, it is \emph{optimally consistent}, or statically optimal, if its consistency is $\Oh(1)$. 
Similarly, a learning-augmented data structure is  $r$-robust iff its total cost 
is at most $r.m$ for any instance. In particular, it is \emph{optimally robust} iff it has robustness $\Oh(\log n)$ for a dictionary with $n$ items.
\end{definition}

Consistency and robustness are not always inherently conflicting attributes, yet achieving one often necessitates compromising the other, as observed in prior research \cite{DBLP:conf/innovations/0001DJKR20,ijcaiKS22,Zeynali0HW21,LiYQSYWL22,sun2021pareto}. Attempts to attain optimal consistency have often resulted in a trade-off that undermines robustness and vice versa. Within the literature \cite{Knuth71,ChenChen2023V2,LinLWood22}, several data structures have been proposed with an emphasis on optimizing consistency, showcasing exceptional performance under ideal predictive conditions. However, this pursuit of high consistency tends to render these structures less robust when faced with adversarial predictions. Conversely, other data structures, such as balanced BSTs and AVL tree, prioritize bolstering their robustness, particularly under worst-case scenarios, thereby offering enhanced resilience. 
 In light of these results, one might ask whether it is possible to get the best of the two worlds: optimal consistency and robustness at the same time. We answer this question in the affirmative in this paper.

\subsection{Treaps and Skip Lists}
A treap is a binary search tree where items has additional field which is its priority. In addition to binary search tree property, a treap follows the heap property.
In a classic {\em random treap}~\cite{SeidelA96}, priorities are assigned randomly, allowing them support dictionary operations in expected $\Oh(\log n)$ time, i.e., they are optimally robust. On the other hand, random treaps are not optimally consistent due to a lack of information about access frequencies.
When frequency predictions are available, it is natural to define priorities based on frequencies rather than randomly. These ideas were explored in~\cite{LinLWood22,ChenChen2023V2} to define learning-augmented treaps (see \S\ref{sect:learningTreaps}). 

Skip lists, introduced by~\citet{Pugh90}, are randomized data structures that provide all functionalities of balanced BSTs on expectation. A skip list is a collection of sorted linked lists that appear in \emph{levels}, where all dictionary items are present in the lowest level, and a subset of items at each level are selected randomly and independently to be replicated in the above layer. The random decisions
follow a geometric distribution, ensuring an expected $\Oh(1)$ replicas per key. Each replica has a ``right pointer'' to the next item in the its linked-list, and a ``bottom pointer'' to the replica of the same item in the level below. The top list is said to have {level} 1, and the level of any other list is the level of its top list plus~1. The ``height'' of skip lists, is defined as the number of levels in it, and ``depth'' of a specific item, is the least level number at which the item appears.

Previous work has established a tight relationship between skip lists and trees. For example, a skip list can simulate various forms of multiway balanced search trees (MWBSTs) such as B-trees, and (a,b)-trees~\cite{MunroPS92,MehlhornN92,BagchiBG05}. On the other hand, \emph{Skip trees}~\cite{Messeguer97} are MWBSTs that certify one-to-one mapping between MWBSTs and skiplists. It is also possible to represent any skip list with a BST~\cite{DeanJ07,BoseDL08}. 

\subsection{Learning-augmented Treaps}\label{sect:learningTreaps}

In the learning-augmented treaps proposed by~\citet{LinLWood22}, priorities are not random but instead predicted frequencies.
Under a {\em random-order rank assumption}, which requires that items' {ranks} to be a random permutation of $[n]$, these treaps are optimally-consistent and have bounded robustness.
However, these treaps are {not optimally consistent} nor robust when the random-order rank assumption does not hold. In particular, when the predicted frequency distribution is highly skewed, the resulting treap may resemble a linear list, which is clearly neither robust nor consistent, as shown in the following proposition.

\begin{proposition}[Appendix~\S\ref{app:learned_treap_performance}]
\label{prop:ICML}
Without random-order rank assumption, the consistency of learning-augmented treap of~\citet{LinLWood22} for dictionaries of size $n$ is at least $\Omega(n /\log n)$, and its robustness is $n$.
\end{proposition}

~\citet{LinLWood22} present a method for relaxing the random-rank assumption by bijecting each key in the dictionary to a random key in a secondary dictionary maintained by a treap based on predicted frequencies.
This simple solution, however,  shuffles the keys, and the resulting tree is not ordered.
Therefore, one cannot efficiently answer secondary dictionary such as predecessor and successor operations using these shuffled treaps.

Recently, ~\citet{ChenChen2023V2} introduced a different learning-augmented treap, in which a random element is introduced in priority assignment. Precisely, the priority of the node with key $i$ is defined as $\texttt{pri}_i = \delta_i - \lfloor \log_2 \log_2 1/\hat{f}_i \rfloor, \ \  \delta_i \sim U(0,1),$  where $U(0,1)$ is the uniform distribution over interval $[0, 1]$.  
The total number of comparisons is proved to be $\Oh(m \cdot \texttt{Ent}(\bm{f},\bm{\hat{f}}))$, where $\texttt{Ent}(\bm{f},\hat{\bm{f}})$ is the \emph{cross entropy} between $\bm{f}$ and $\hat{\bm{f}}$ defined as $\texttt{Ent}(\bm{f},\hat{\bm{f}}) = - \sum_{i\in[n]} f_i \log \hat{f}_i$. When $f_i = \hat{f}_i$, it holds that $\texttt{Ent}(\bm{f},\hat{\bm{f}}) = H(\bm{f})$, and the total number of comparisons will be $\Oh(m H(\bm{f}))$, {ensuring} optimal consistency  for these treaps. However, when predictions are adversarial, these treaps are not robust, as shown in the proposition below.

\begin{proposition}[Appendix~\S\ref{app:chenchen_proof}]
\label{prop:chenchen}
    The learning-augmented treaps of~\cite{ChenChen2023V2} have optimal consistency (are statically optimal), but {their robustness is $n$}.
\end{proposition}

Intuitively, regardless of the amount of randomness (noise) added to the priorities, 
an adversary can skew the predicted frequency distribution to negate the added random noise. For that, it suffices to define $\hat{f}_i$'s in a way to ensure ${\hat{f}_i \geq \hat{f}_{i-1} + u}$, where $u$ is the upper bound for the random variable added to the priority.

In other words, regardless of the priority assignment, the treap will be either fully random by ignoring predictions, hence non-consistent, or not robust against adversarial predictions. \textit{This suggests that, in order to achieve optimal consistency 
(static optimality) and robustness simultaneously, one must consider data structures other than treaps.}

\section{Consistent and Robust Dictionaries}\label{sect:mainDic}
 This section presents our main results on learning-augmented dictionary data structures: a skip list that achieves optimal consistency (statically optimal) when predictions are correct and offers robustness of $\Oh(\log n)$ when predictions are adversarial. 

 If one wants to maintain a \emph{static dictionary}, without the support of insertions and deletions, it is rather easy to achieve a consistent and robust data structure as follows. Provided with predicted frequencies $\hat{\bm{f}}$, first, build the static optimal tree of \citet{Knuth71} in $\Oh(n^2)$; call the resulting tree \emph{optimistic tree} $T_o$. Also, form a balanced BST $T_o$, referred to as \emph{pessimistic tree}. To access an item with key $i$, we simultaneously search for $i$ in $T_o$ and $T_p$. When one tree finds the item, we immediately stop the searching process of another tree. The number of comparisons would be $2\min(\text{Search}_{T_o}(i), \text{Search}_{T_p}(i))$, ensuring both static optimality and optimal robustness of $O(\log n)$.

However, this simple data structure does not support insertion and deletion queries, as Knuth's static optimal tree lacks support for these operations. In practice, the predicted frequencies get updated frequently, and such updates are implemented by deletion and insertion, necessitating a \emph{dynamic dictionary that supports insertions and deletions}. Moreover, the time complexity of constructing the optimistic tree is expensive, rendering it impractical for widespread use.

In what follows, we first present \DSName in Section~\ref{sec:robustSL}. In Section~\ref{sec:theory}, we show that \DSName achieves optimal consistency and robustness. Last, in Section~\ref{sec:ext}, we discuss extensions, such as how \DSName can be augmented to efficiently support secondary operations like the predecessor, successor, rank, select, and range queries.

\subsection{\DSName, Consistent and Robust Skip list}
\label{sec:robustSL}
\DSName is a skip list, and many of its functionalities are similar to a regular skip list. In particular, upon insertion of an item with key $i$, it replicates the item in a certain number of levels (each associated with a linked list) starting at the lowest level.
Unlike a regular skip list, where the replication strategy is purely randomized, \DSName involves predicted frequencies in its replication strategy. As we will show, the number of replicas for item with key $i$ in \DSName is a function of its predicted frequency and a geometric random variable.
The key to achieving optimal consistency and robustness in \DSName lies in the precise classification of items based on their predicted frequencies.
The goal is to assign more replicas (lower ``depths'') to items with higher predicted frequencies, while items with similar predicted frequencies share the same expected depth, ultimately reinforcing the consistency of \DSName. Importantly, \DSName ensures that maximum number of comparison made while searching any item is at most $\Oh(\log n)$, ensuring optimal robustness.

Before presenting \DSName formally, we present some intuitions behind it via a simpler data structure that illustrates the main ideas behind \DSName.

\paragraph{High-level intuitions and ideas.}

Let's assume a number of items in the dictionary, $n$, is fixed (this assumption will be relaxed later in Section~\ref{sec:ext}).
We classify items such that items with predicted frequency larger than 1/2 belong to the first class (index 0), those with predicted frequency in $(1/4,1/2]$ belong to class index 1, and more generally, items with predicted frequency in $(1/2^{2^{c+1}}, 1/2^{2^{c}}]$ belong to class $c$. We further limit the number of classes into $K = \lceil \log \log n \rceil$; that is, items with frequency smaller than $1/n$ belong to class $K$. Maintain a separate skip list for items of each class, and, to search for any item with key $i$, first examine the skip list maintained for the first class, and in case of not finding $i$, we examine the skip list for classes $1,2,\ldots, K$ in the same order.

We show that searching for an item in class $c$ takes no more than $\alpha \cdot 2^{c+1}$ comparisons on expectation for some constant $\alpha$. To see that, let $N_c$ denote the number of items in class $c$.
Observe that at most $2^{(2^{c+1})}$ items belong to class $c$ (i.e., $N_c \leq 2^{(2^{c+1})}$); otherwise, the total predicted frequencies for items in class $c$ exceeds 1.
Therefore, searching for an item in the skip list of class $c$ takes the expected time of $\alpha\log N_c \leq \alpha 2^{c+1}$. Now, if $i$ is assigned to class index $c$, the total expected number of comparisons for finding $i$ would be $\sum_{c'\leq c} \alpha 2^{c'+1} < \alpha 2^{c+2}$. 
The robustness of the data structure is direct from the fact that $c \leq K = \lceil \log \log n\rceil$, and thus searching any item takes at most $4\alpha2^K = O(\log n)$. For consistency, when predictions are accurate,
we note that if an item with key $i$ is assigned to class $c$, then it holds that ${f_i} = {\hat{f}_i} \leq {2^{-2^{c}}}$, which is direct from the way classes are defined. Thus, the expected number of comparisons for finding $i$ would be $\alpha \cdot 2^{c+2} \leq -4\alpha \cdot \log {\hat{f}_i}$.
We can conclude that searching for $i$ takes $O(-\log f_i)$, proving our algorithm's consistency. 

\DSName improves the above data structure in a few ways. First, it maintains a \emph{single} skip list for all classes, which is necessary for efficient handling of frequency updates: when an item's frequency is updated, it is desirable to update the number of replicas in a single skip list rather than removing them from one skip list and adding to another. To maintain a single skip list, we ensure that items of class $c$ form a ``layer'' above those of class $c+1$; this is achieved via defining a ``class-base'' for each class, ensuring a minimum number of replicas for class $c$. Second, instead of using powers of $1/2$ in the classification, \DSName uses a parameter $\pz$, which is set to optimize the constant for static optimality.

\paragraph{Algorithmic details of \DSName.}
Similarly to the data structure described above, \DSName classifies items based on their predicted frequencies and employs a two-step process to select the number of replicas for each item. Firstly, it replicates any item in class $c$ a minimum of $\mathcal{D}(K) - \mathcal{D}(c)$ times, where $\mathcal{D}(c)$ is an increasing function of $c$, indicating the maximum possible depth of items in class $c$. Given this definition, $\mathcal{D}(K) - \mathcal{D}(c)$, called \BaseClass{c}, determines the minimum number of replicas for items in that class. Intuitively, items in class~0, which are predicted to appear more in the input, have a higher class-base, meaning that they are replicated more than other classes; as the indices of classes increase, their class-base decreases, implying less replication. See Figure~\ref{fig:rsl_structure} for an illustration.

In addition to replicas specified by the class-base, more
replications are introduced for each item, guided by a geometric distribution with a parameter of $1-\pg$, where $\pg < 1$ is a constant value. In the next section, we present the consistency and robustness of \DSName as a function of $\pg$, and show how this parameter impacts those metrics. This stochastic process for adjusting the number of replicas ensures that the expected search cost within a class is logarithmic to the number of items in that class. The following sections explore the classification and replication processes within \DSName in detail.

\begin{figure}[!t]
    \centering
    \includegraphics[width=0.7\linewidth]{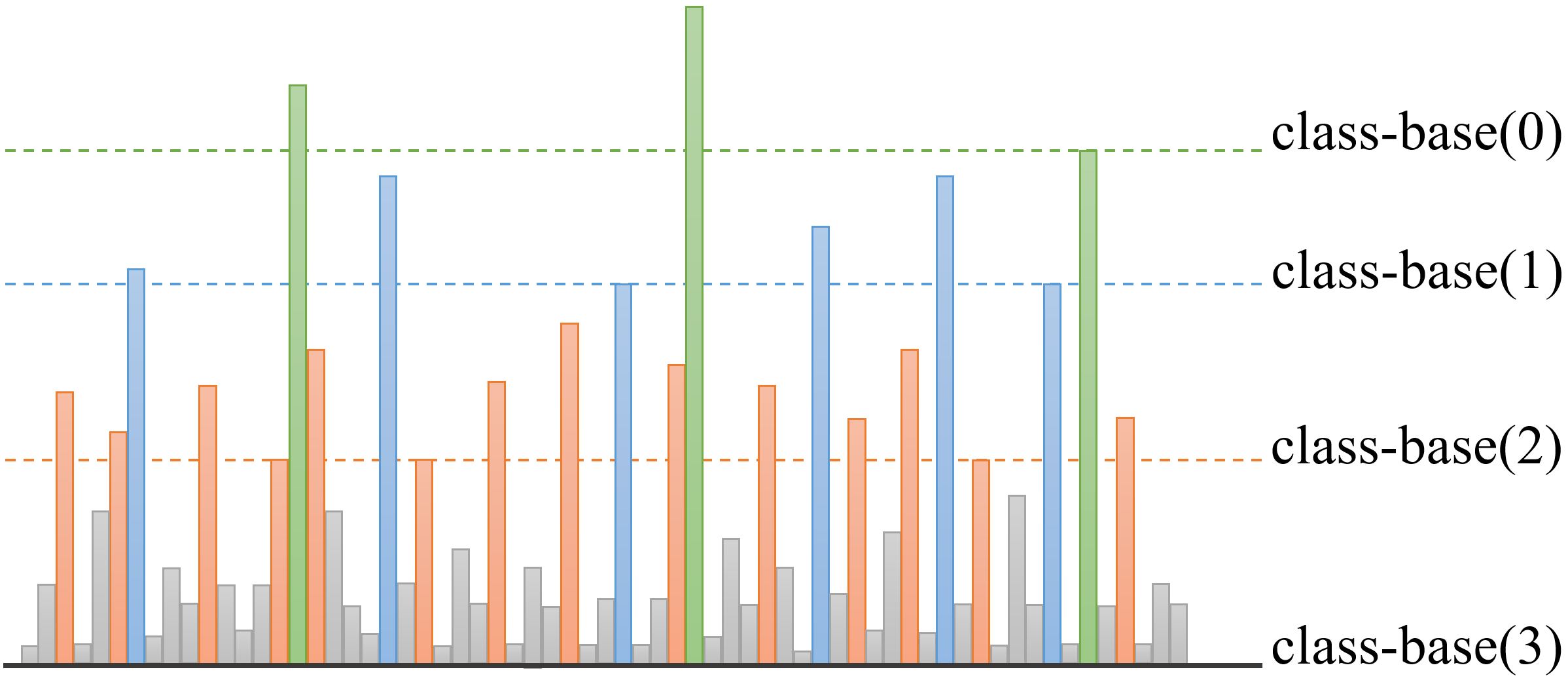}
    \caption{Structure of \DSName having four classes. Items in lower-indexed classes replicate more. The replication of items within a class is achieved through a stochastic process.}
    \label{fig:rsl_structure}
\end{figure}

\textit{Classification approach.} The classification of items relies on their predicted frequencies. An item $i$ with a predicted frequency $\hat{f}_i$ is assigned to class $c_i (\geq 0)$ as follows. First, items with predicted frequency $\geq \theta$  belong to class $0$. Other items belong to class $c_i \geq 1$, if their predicted frequency satisfies the condition: 
\begin{equation}
\label{eq:classification} 
\pz^{2^{c_i}} \leq \max(\hat{f}_i, n^{\log(\pg)/2}) < \pz^{2^{c_i - 1}},
\end{equation}
which gives
\begin{equation*}
    c_i = \lceil \log\bigg( \frac{-1}{\log \pz}\min\bigg(-\log \hat{f}_i, \frac{-\log p \log n }{2} \bigg) \bigg) \rceil,
\end{equation*}
 
where  $\pz \in (0,1)$ is an algorithm parameter. 
In particular, the largest value of $c_i$ is realized for items with small predicted frequency, where $-\log \hat{f_i} \leq -\log p \log (n)  /2$. These items (if exist) belong to the class index denoted by $K$, where $$K = 1 + \lceil
 \log \log n - \log(2 \log \pz / \log p )\rceil.$$
Intuitively speaking, parameter $\pz$ determines the predicted frequency range for items in the same class. A smaller value of $\pz$ widens the range, which implies fewer classes, while higher $\pz$ narrows it, resulting in more classes. Regardless, the number of classes is bounded by $K+1 = \Theta(\log \log n)$.

\textit{Number of replicas.}
To determine the number of replicas for an item upon its insertion, \DSName first calculates the class-base of the item, which is deterministically defined based on the class of item, $c$, and is calculated as $\mathcal{D}(K)-\mathcal{D}(c)$. We recursively define $\mathcal{D}(c)$ as follows. First, for $c=0$, we let
$\mathcal{D}(0) = \lceil \log \pz /\log \pg \rceil$, and
for any $c > 0$, we define:
\begin{align}
\mathcal{D}(c) = \mathcal{D}(c - 1) +  \lceil \frac{\log \pz}{\log \pg}2^c \rceil.
\end{align}
Intuitively, $\mathcal{D}(c)$ represents the maximum depth of items belonging to class $c$, and the above definition ensures that items in classes larger than $c$ appear at least  ${2^c\log \pz}/{\log \pg} $ levels deeper than items of class $c$, on expectation. 
The number of  replicas for an item $i$ in class $c_i$ is {calculated as}:
\begin{align}
\label{eq:rep_i}
{h_i = \mathcal{D}(K) - \mathcal{D}(c) + \mathcal{X}(\pg),}
\end{align}
where $\mathcal{X}(\pg)$ is a random variable following a geometric distribution with parameter $1 - \pg$. For example, items of class $K$ (if there is any), which are the ones with the smallest predicted frequencies, are replicated only based on stochastic replications, $\mathcal{X}(\pg)$, ensuring that they are at the deepest levels of the skip list. As another observation, for any $c < c'$, items of class $c$ are replicated in $D(c')-D(c)$ more levels, on expectation, than items of class $c'$. This is consistent with the higher predicted frequencies for items of class $c$. Finally, note that items that belong to the same class all have the same expected number of replicas, ensuring that the data structure resembles a regular skip list for items within the same class.

\subsection{Theoretical Analysis of \DSName}
\label{sec:theory}

In what follows, we analyze \DSName, providing upper bounds on the number of comparisons made when accessing any item. In particular, we establish the optimal consistency and robustness of \DSName. Let $c \in [0,K]$ be any class defined by \DSName; we use $N_c$ to denote the number of items belonging to class $c$. 
If $ c \leq K-1$, then for any item $i$ in class $c$, it holds that $\hat{f}_i \geq \theta^{2^c}$. Consequently, it follows that $N_c \cdot \theta^{2^{c}} \leq 1$, as violating this condition would contradict the requirement that predicted frequencies form a probability distribution ($\sum_i \hat{f}_i = 1$).
Thus, we deduce that $N_c \leq \theta^{-2^{c}}$, leading to the inequality $\log N_c \leq -2^{c} \log \theta$.

The following lemma presents an analysis of the classification and replication mechanisms within \DSName, illuminating that beyond the inherent relationship between higher predicted frequencies and greater heights, items belonging to distinct classes are anticipated to possess varying numbers of replications.

 \begin{lemma}[Appendix~\S\ref{app:proof_low_prob}]
    \label{lem:cond_low_prob}
    For any class index $c > 0$, the expected number of items of class $c$ that are replicated at least \BaseClass{c-1} times is less than 1.
\end{lemma}

 This lemma indicates that, while searching for a key $i$ categorized within a specific class, $c_i$, the expected number of comparisons conducted with items in higher-indexed classes, $c_j > c_i$, is significantly lower than the expected comparisons with items in class $c \leq c_i$. This property helps \DSName to access more frequent items with a lesser number of comparisons. Using this lemma, we can provide an upper bound from the number of comparisons made while searching for an item $i$.

\begin{lemma}[Appendix~\S\ref{app:proof_shahinMain}]\label{lem:shahinMain}
The expected number of comparisons made while searching an item $i$ in class $c_i < K$ is at most $\frac{4}{p \log{p}} \log(\hat{f}_i)$. For an item $i$ in class $c_i = K$, the expected number of comparisons is at most $\frac{4\log \theta-1}{p \log p} (\log n) + \mathcal{O}(1) $.
\end{lemma}
The above lemma bounds the cost of access to any item based on their predicted frequencies. In the following, we establish 
the consistency of \DSName, which follows from Lemma~\ref{lem:shahinMain}, noting the number of comparisons for an item $i$ of frequency ${f}_i = \hat{f}_i$ is proportional to $\log(f_i)$ for all classes, including class $K$.

\begin{theorem}[Appendix~\S\ref{app:proof_rsl_consistency}]
\label{thm:rsl_consistency}
    The consistency of \DSName can be established as $\frac{-4}{\pg\log(\pg)} = \Oh(1)$, ensuring the static optimality for \DSName.
\end{theorem}

The following lemma is direct from Lemma~\ref{lem:shahinMain}, by noting that the maximum number of comparison made while searching for items belong to class $K$.
\begin{theorem}[Appendix~\S\ref{app:proof_rsl_robustness}]
    \label{thm:rsl_robustness}
    The maximum cost of searching for items within \DSName is $\Oh(\log n)$, which makes \DSName optimally-robust.
\end{theorem}

\begin{figure*}[!t]
	\centering
    \subfigure{\includegraphics[width=0.96\textwidth]{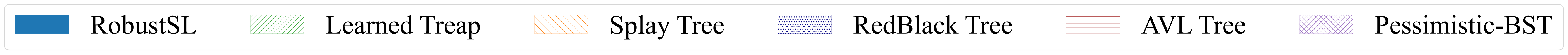}}\vspace{-4mm}
    \subfigure{\includegraphics[width=0.32\textwidth, height=3.3cm]{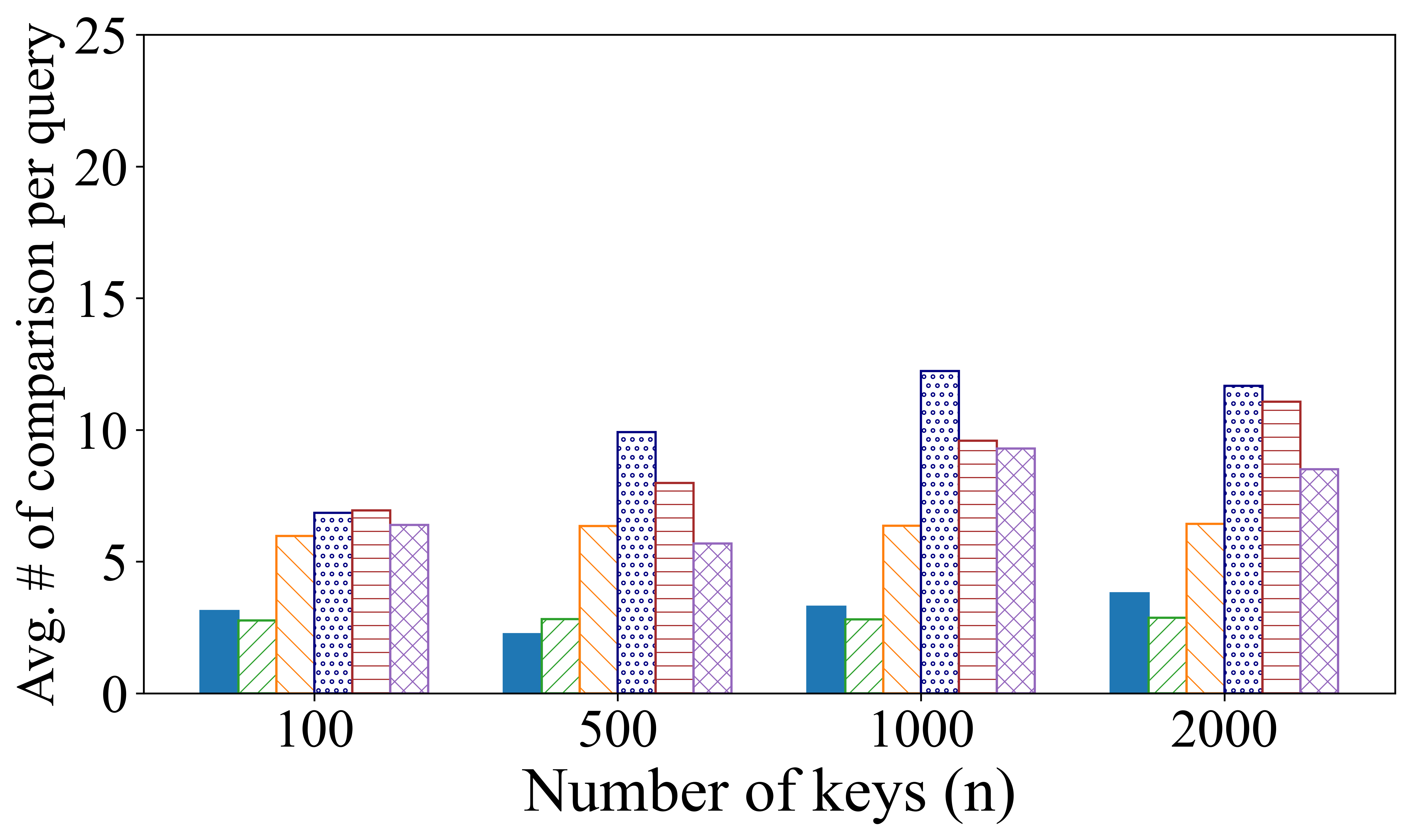}}\hspace{2mm}
	\subfigure{\includegraphics[width=0.32\textwidth, height=3.3cm]{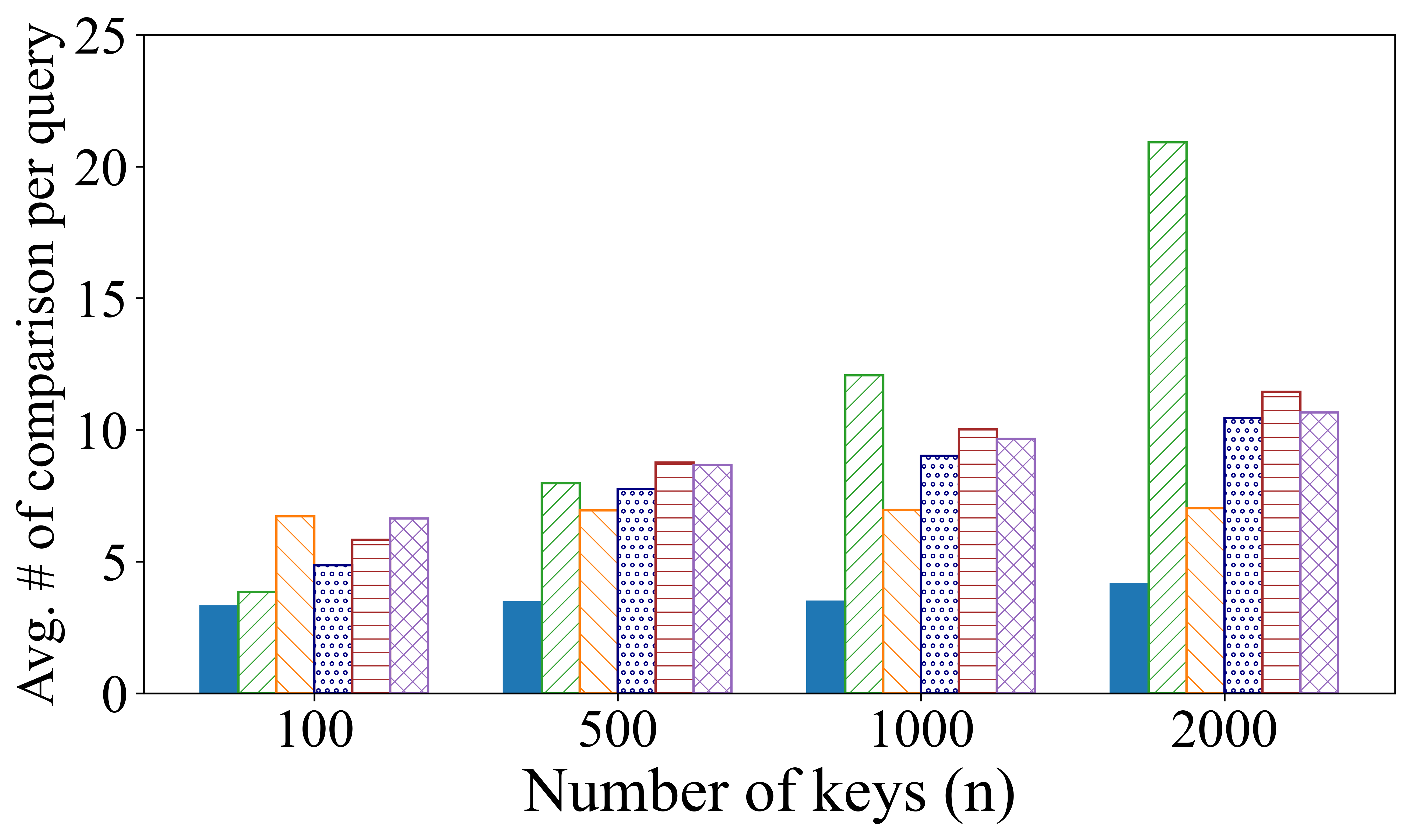}}\hspace{2mm}
    \subfigure{\includegraphics[width=0.32\textwidth, height=3.3cm]{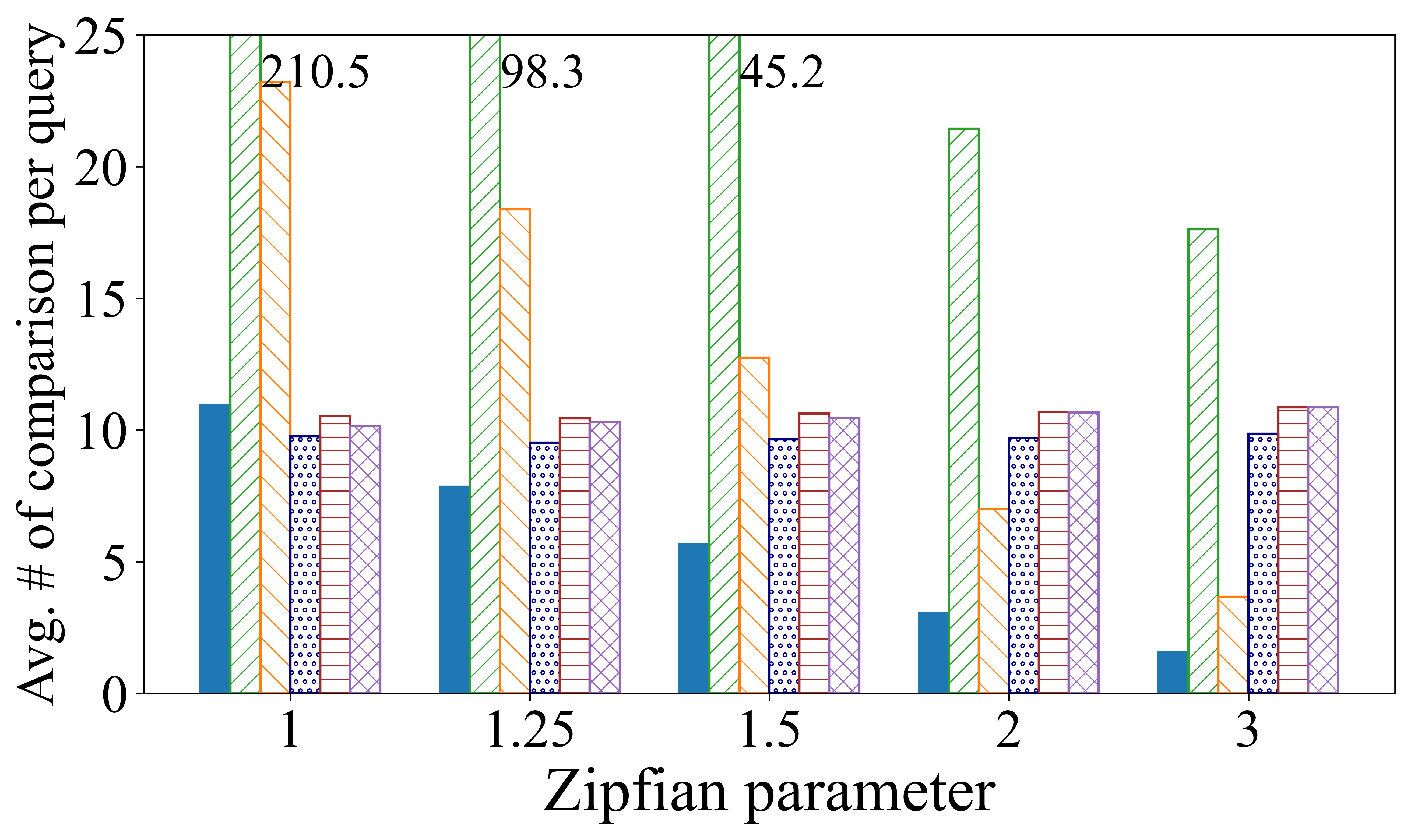}}
	\caption{Average number of comparisons per query of \DSName and baseline data structures for dynamic evaluations under (left) random frequency ordering with perfect predictions, (center) adversary frequency ordering with noisy predictions, and (right) adversary frequency ordering with different value of Zipfian parameter. The performance of learned treap is highly impacted by the prediction error, while \DSName shows its robustness against noisy predictions. To improve visibility, the y-axis range in the right figure is limited to 25, with actual values displayed next to bars exceeding this threshold for clarity.}
	\label{fig:DynamicExp}
\end{figure*}

\subsection{Discussion and Extensions}
\label{sec:ext}
\paragraph{Dynamics of \DSName.}

The results in the previous section assume a fixed value of $n$. The actual number of items, denoted by $n_t$, however, is clearly impacted by the insertions and deletions. We explain how to maintain the described consistency and robustness, using a value of $n$, which approximates $n_t$ and 

represents the parameter utilized by \DSName for classification (as depicted in Equation~\eqref{eq:classification}). 
We maintain a value of $n$ that is always larger than $n_t$. 
At the beginning, we set $n = 4$. After any insertion operation, if $n_t$ becomes equal to $n$, we update $n \leftarrow n^2$. After a deletion operation, if $n_t =  n^{1/4}$, we adjust $n \leftarrow \max(4, \sqrt{n})$.

This approach does not compromise the consistency of \DSName, as $n_t \leq n$, preserving the analysis outlined in Theorem~\ref{thm:rsl_consistency}. Additionally, $n$ remains bounded by $n < n_t^4$ for any $n_t > 1$, ensuring $\log n = \Oh(\log n_t)$ consistently. Consequently, the robustness of \DSName maintains at $\log n = \Oh(\log n_t)$.

By employing this technique, the amortized cost of insertion or deletion operations remains $\Oh(\log n)$ (finding the index of a key in the sorted ordering of all key values requires $\Omega(\log n)$ number of operations), as updates to \DSName occur only after at least $n^2$ consecutive insertions or deletions. It's important to note that the maximum cost for entirely reconstructing \DSName, given its maximum depth as $\Oh(\log n)$, is capped at $\Oh(n\log n)$.

\paragraph{Secondary queries.}
The order structure of the skip list allows efficient answering of secondary queries with little augmentation. 
First, we augment the deepest level of \DSName to make it doubly linked. 
We also add a pointer from any node $x$,  ${deep(x)}$, to the node with the same key as $x$ at the deepest level. 
This allows implementing \texttt{predecessor$(i)$}/\texttt{successor$(i)$} (keys before/after $i$ in the sorted order) by searching for $i$, finding it at some node $x$, following ${deep(x)}$, and probing a left/right pointer; the time complexity will be similar to the search.  
Similarly, \texttt{range}($i,j$) (reporting all keys $k$ s.t. $i\leq k \leq j$) can be done by searching for $i$ and following pointers in the deepest level. 
For other secondary queries, 
let $pred(x)$ be the node before $x$ at the same level. 
We let $x$ have an extra integer field $s(x)$ that specifies the number of keys that are ``skipped'' by following the right pointer from $pred(x)$ to $x$, that is, the number of keys in the dictionary between the key of $pred(x)$ and $i$. These augmentations allow efficient answering of \texttt{rank}($i$) (the index of $i$ in a sorted ordering of keys) by searching for $i$ and summing the values of $s(x)$ over nodes on the search path. Similarly, \texttt{select}$(t)$ (the $t$'th smallest key) can be done in $\mathcal{O}(\log n)$, summing over values $s(x)$ by following the right pointers to reach $t$. 
\begin{proposition}
    It is possible to augment \DSName to answer \texttt{predecessor}$(i)$, \texttt{successor}$(i)$, \texttt{rank}$(i)$, all in time proportional to $\texttt{Search}(i)$, \texttt{range}($i,j$) in time proportional to $\texttt{Search}(i) +z$, where $z$ is the output size, and \texttt{select}($t$) in $\Oh(\log n)$.
\end{proposition}

\begin{figure*}[!t]
	\centering
    \subfigure{\includegraphics[width=0.96\textwidth]{figures/Header.png}}\vspace{-4mm}
    \subfigure{\includegraphics[width=0.32\textwidth, height=3.3cm]{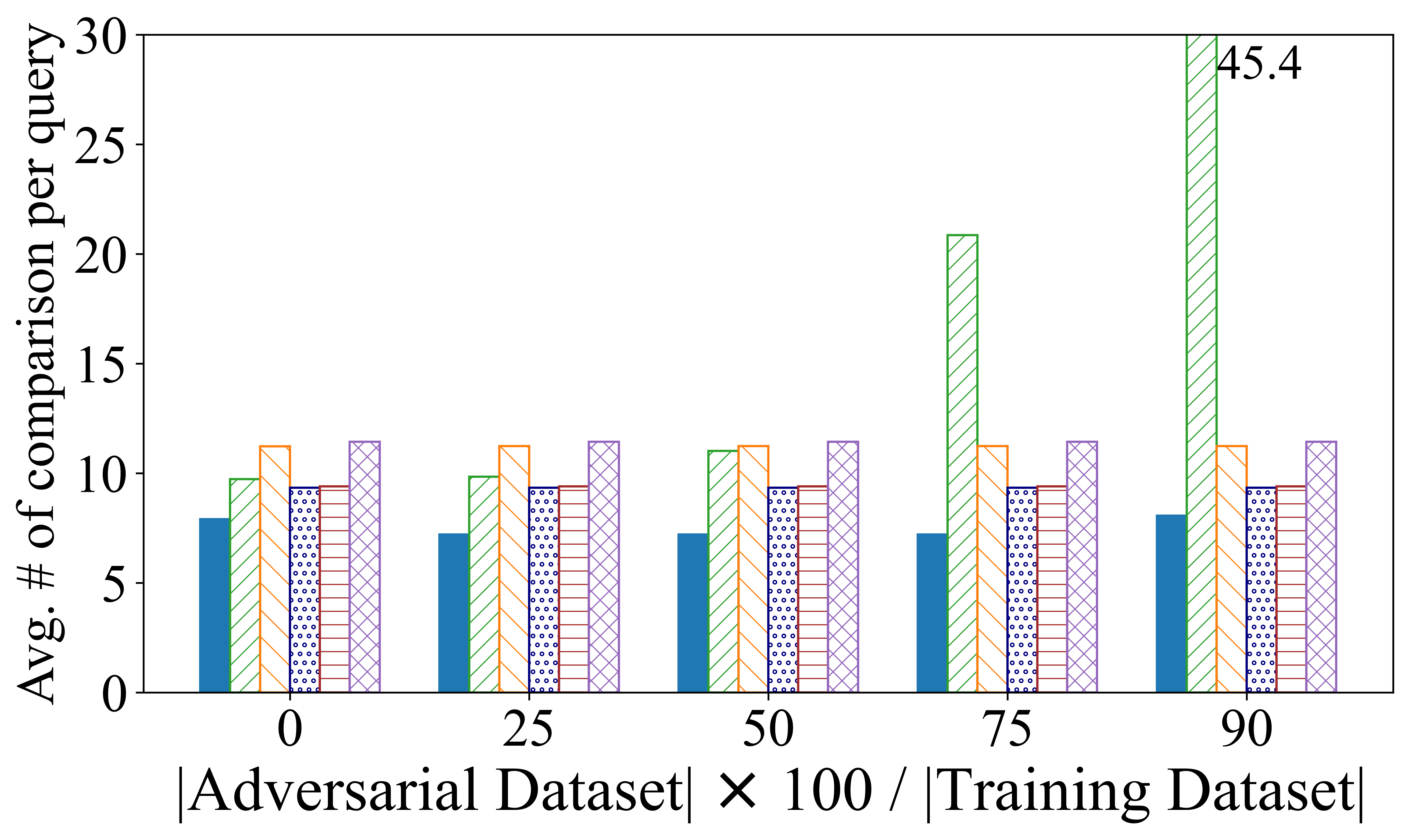}}\hspace{2mm}
    \subfigure{\includegraphics[width=0.32\textwidth, height=3.3cm]{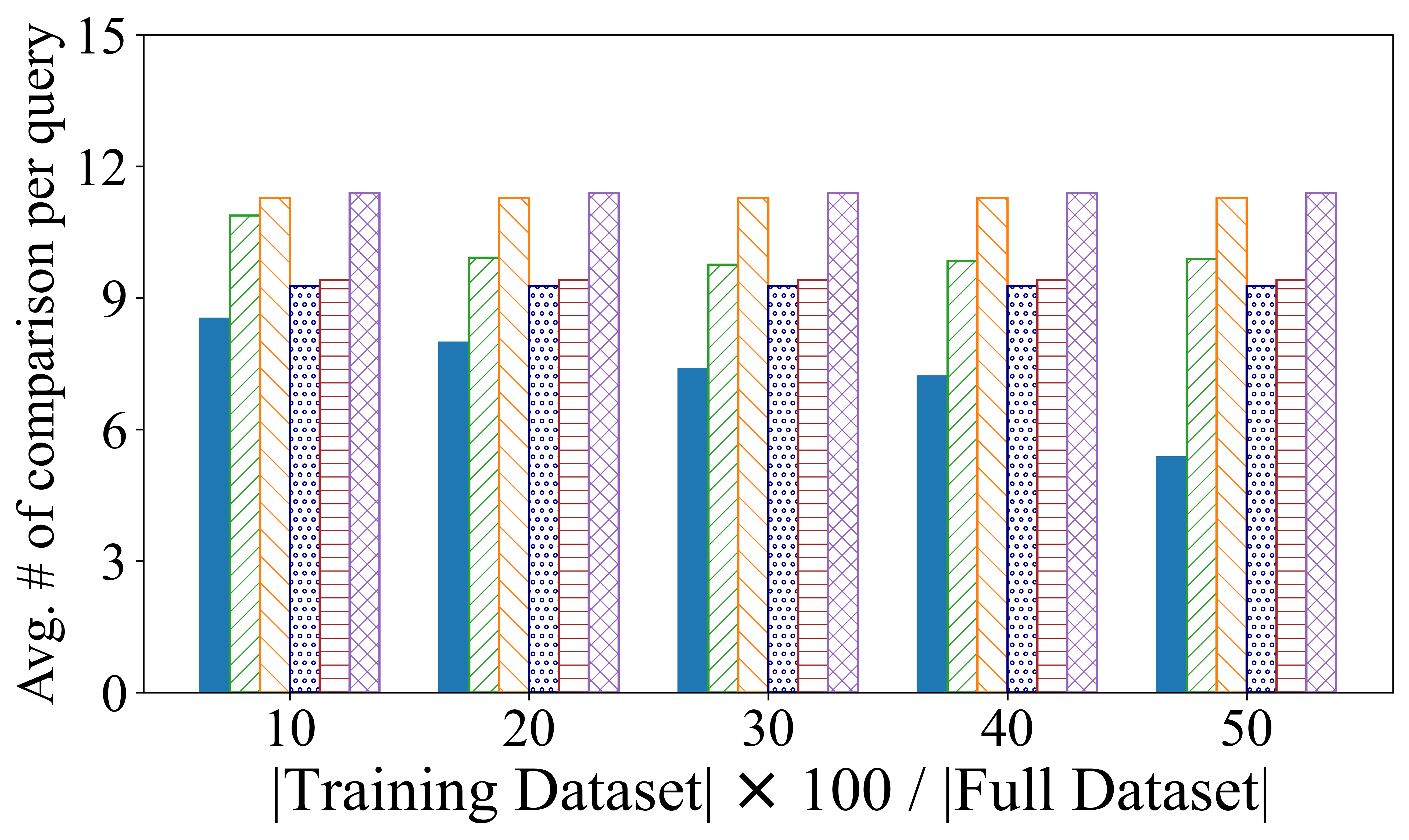}}\hspace{2mm}
	\subfigure{\includegraphics[width=0.32\textwidth, height=3.3cm]{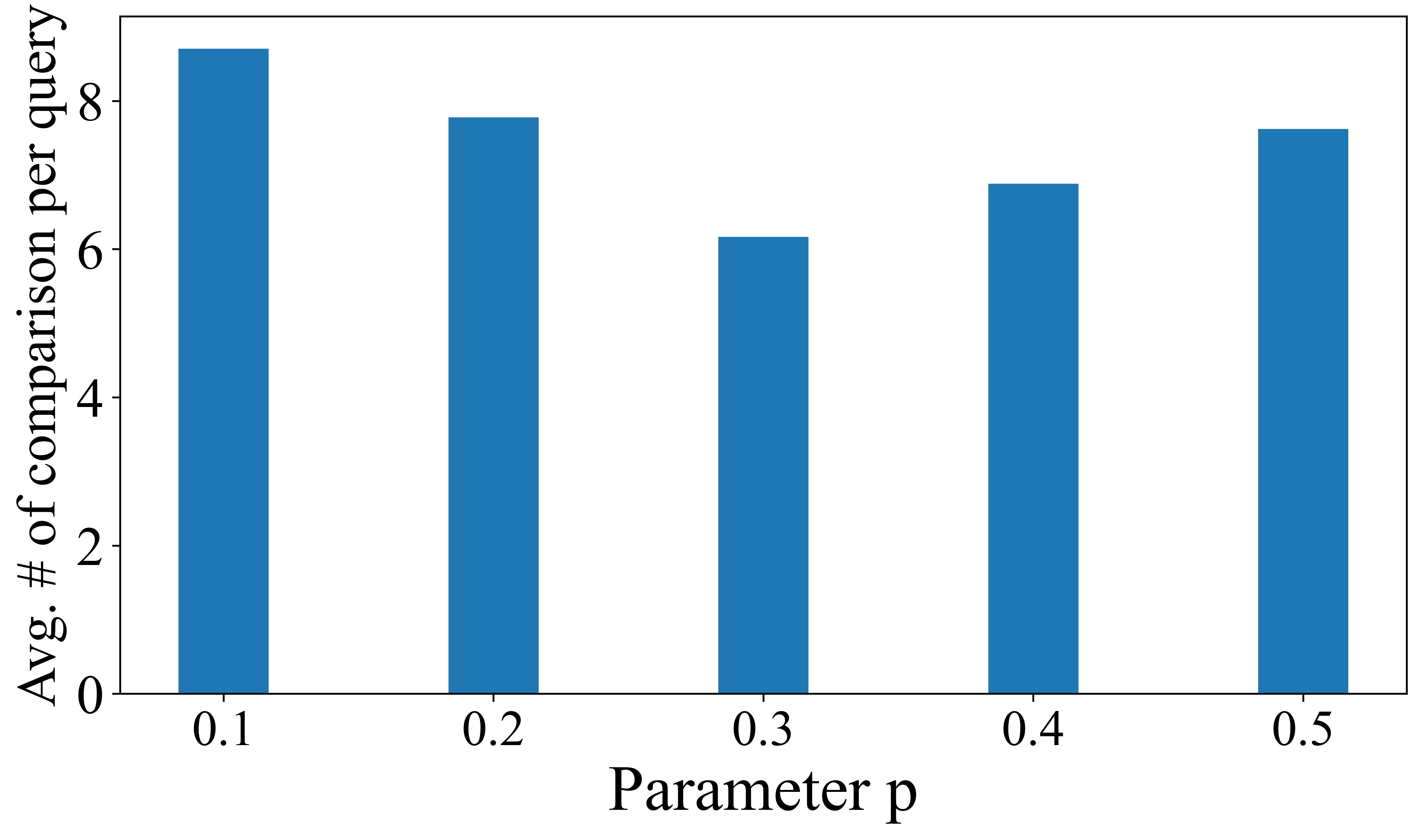}}
	\caption{Average number of comparisons per query of \DSName and the baseline data structures. The evaluation is conducted by varying three factors: (left) sizes of the adversarial dataset, (center) sizes of the training dataset, and (right) value of \pg in \DSName.  Notably, the performance of learning-augmented treaps demonstrates a significant increase when predicted frequencies are adversarial. }
	\label{fig:BBCExp}
\end{figure*}

\section{Experiments }\label{sect:exp}
 In this section, we evaluate the performance of \DSName in both static and dynamic settings and compare the performance with alternatives. Our goal is to investigate the robustness and consistency of \DSName under perfect and adversarial predictions. 

\subsection{Experimental Setup and Overview}
\label{sec:exp_setup}
We compare \DSName with several alternative data structures, namely AVL trees, red-black trees, splay trees, balanced BST (pessimistic BST), and learning-augmented treaps of~\citet{LinLWood22}.
We thoroughly evaluate the robustness and consistency of \DSName in static and dynamic settings by considering both perfect predictions and adversarial instances. Specifically, we conduct our experiments in two categories:
(1) dynamic random orders with perfect predictions, and (2) dynamic adversarial orders with noisy predictions. In addition, we provide further experiments concerning static dictionaries in Appendix~\ref{app:exp}. 
In the random order experiments, the frequency ranking of items are randomly selected concerning the key value ranking of items. In addition, we consider a perfect prediction in random order experiments to compare the \textit{consistency} of \DSName and baseline algorithms. On the other hand, the adversarial ranking is used to evaluate the \textit{robustness} of \DSName as compared to alternatives and includes perfect and noisy predictions. For each category, we conduct a series of 100 trials and report the average number of comparisons for a search query of \DSName and baseline data structures. More details on the prediction and error model are given in Appendix\S~\ref{app:exp}. 
We select $\pz = 0.05$ in \DSName for our experiments since it leads to slightly better performance, even though \DSName is robust to the selection of this parameter (as shown in Figure~\ref{fig:pz_impact} in Appendix\ref{app:exp_pz}). In addition, we select $\pg = 0.368$ that minimizes the consistency value presented in Theorem~\ref{thm:rsl_consistency}. 
Last, we use the following two categories of synthetic and real datasets.

\paragraph{Synthetic dataset.}
We conducted experiments to evaluate the performance of \DSName and baseline data structures using synthetic datasets of varying sizes. Specifically, we selected the number of unique keys, $n$, from the set of values $[100, 500, 1000, 2000]$. For each value of $n$, we generated $m = 100,000$ search queries with each item appearing according to the Zipfian distribution with a parameter of $2$. In addition, we test the performance of \DSName and other data structures against Zipfian distribution with different parameters. For random order experiments and experiments on testing different values of Zipfian distribution, we fix the number of items to $2000$. Our results demonstrate that the performance of learning-augmented treaps is much worse against the adversary frequency ordering using low-accuracy predicted frequencies. For a fair comparison between the performance of \DSName and the baselines data structures, in experiments with adversarial frequency predictions, we select the number of trials with high-accuracy predictions up to $99$ times of low-accuracy prediction trials. This means that in the experiments with ``adversarial'' frequency predictions, $99$ trials used high-accuracy predictions, and a single trial was designed using low-accuracy predictions (see Appendix \S\ref{app:exp} for details).

\paragraph{BBC news article dataset.} We also use the BBC news article dataset~\cite{KaggleBBC} to evaluate the performance of
\DSName and other baseline data structures in responding to news article queries~\cite{chenXi2022,kostakos2020strings}. In these experiments, we select a fraction of the entire dataset to predict item frequencies (training dataset) in the remaining portion (test dataset). To evaluate data structure performance under noisy conditions, we artificially generated additional articles using an adversarial approach. These adversarially-generated articles were designed to align ranking of frequencies with keys. In addition we test the impact of parameters $\pg$ and $\pz$ on the performance of \DSName against this dataset. We randomly select $40\%$ of the entire data as the training dataset, $25\%$ (of the training dataset) as the size of the adversary dataset, $\pz = 0.05$, and $\pg = 0.368$ (as used in synthetic experiments) and conducted experiments by varying one factor during each experiment while other factors remained fixed. Our analysis considers the top 5500 items with the highest frequencies for dictionary generation, and search queries were exclusively conducted on these items.

\subsection{Experimental Results}

\paragraph{Synthetic results.} Figure~\ref{fig:DynamicExp} presents the results of our experiments on dynamic structures using synthetic data, showcasing the average number of comparisons per query for both \DSName and baseline data structures under random and adversary settings. While the performance of learning-augmented treaps is notably impacted by error rates, \DSName demonstrates consistent and robust performance, validating its theoretical resilience. Specifically, when searching a dictionary of size $2000$, under random frequency ordering, the average number of comparisons of \DSName is $24.5\%$ more than of learning-augmented treaps under random frequency ordering while \DSName achieves $404.7\times$ lower comparison per query when the predicted frequencies were adversarial. Notably, we repeated the adversarial experiments with a higher number of items (e.g., $n = 1000$), the average number of comparisons per query slightly increased for all dictionaries, while for learned treaps, it increased to $788$, which confirms the vulnerability of these data structures to errors in prediction.
Finally, we replicated the same experiment in a static setting, and since the results were very similar to those obtained in the dynamic setting, we have included the results and additional analysis on static experiments in Appendix~\S\ref{app:exp}.

\paragraph{Dataset results.}
Figure ~\ref{fig:BBCExp} illustrates the average number of comparisons per query for the compared data structures, varying size of adversarial and training datasets, and $\pg$ in \DSName. Results align with theoretical analyses and synthetic dataset evaluations.
Learned treaps perform well when the prediction error is low (zero-size adversarial dataset), but their performance significantly deteriorates under fully adversarial conditions (i.e., when the size of the adversarial dataset is comparable to the training dataset). In contrast, \DSName demonstrates consistent performance and robustness throughout these experiments. When testing the impact of the training sample size, \DSName shows significant improvement in its performance (smaller number of comparisons) with a larger training dataset, while this improvement for learned treaps was much lesser.
Finally, {testing} the impact of $\pg$ shows that the performance of \DSName is optimized when $\pg$ is around $0.3$ supporting theoretical results. It is worth mentioning that our analysis reveals the impact of $\pz$ on the performance of \DSName is negligible. The result of this experiment is given in Appendix~\S\ref{app:exp}.

\section{Conclusion}

In this paper, we presented \DSName, a skip-list-based data structure 
that achieves optimality with quality predictions and stays robust with adversarial prediction.
As a prospect for future exploration, adjusting \DSName's structure to support the availability of partial information, particularly when frequency predictions are available only for a subset of items, presents an interesting avenue for future work.
Whether the guarantees provided by \DSName concerning consistency and robustness can be achieved with any learning-augmented BST is another open question for future study.

\section*{Acknowledgments}
The work of Ali Zeynali and Mohammad Hajiesmaili was supported by the U.S. National Science Foundation (NSF) under grant numbers CAREER-2045641, CPS-2136199, CNS-2106299, and CNS-2102963. Shahin Kamali was supported by the Natural Sciences and Engineering Research Council of Canada (NSERC) under grant number  DGECR-2018-00059.

\bibliographystyle{plainnat}
 \bibliography{references.bib}  




\newpage
\appendix
\onecolumn
\section{Proofs of Theoretical Result}
\label{app:proofs}
\subsection{Proof of Proposition~\ref{prop:ICML}}
\label{app:learned_treap_performance}
\begin{proof}
Consider defining the predictions $\bm{\hat{f}}$ in a way that all items are predicted to have a frequency that is almost $1/n$ with a small additive factor that ensures item $i$ has rank $i$. More precisely, frequency of access to item $i$ is $\hat{f_i} = 1/n + \epsilon_i$ ($\epsilon_i$ could be negative), where $\epsilon_i < \epsilon_j$ for any $i < j$, and $\sum_i \epsilon_i = 0$ and $|\epsilon_i| \ll 1/n$. Given that the key of each item equals its predicted rank, the learning augmented treap resembles a single path, and its height and robustness are linear to $n$. Even when predictions are accurate, treap's total cost for accesses to item $n$, which is $m \hat{f_n}$, is at least $ m/n$, and the total cost for item $i$ is $m \hat{f_i} = m \cdot i (1/n + \epsilon_i)$. The amortized cost for a single request is thus lower bounded by $\sum\limits_{i=1}^n i\ (1/n+\epsilon_i) =  \Omega(n).$
The cost of the treap for $m$ access operations is thus $\Omega(mn)$, while the entropy of $\bm{\hat{f}}$ is $O(\log n)$. We can conclude that the consistency of the learned treap is $\Omega(n /\log n)$.
\end{proof}

\subsection{Proof of Proposition~\ref{prop:chenchen}}
\label{app:chenchen_proof}
\begin{proof}
    Optimal consistency is proven in [Theorem 4.8]~\cite{ChenChen2023V2}. To provide a lower bound for robustness, consider a highly skewed distribution for the predicted frequency predictions such that $\log \log \hat{f}_i \geq \log \log \hat{f}_{i-1}+1$ for any $i\in [n]$.
As a result, the random component of the priority, which is in (0,1), does not make any difference in the rank of items in the resulting tree. That is, the item with key $i$ will have rank $i$, and the learned treap will be highly unbalanced, resembling a linear list and therefore, its robustness is $n$.
\end{proof}

\subsection{Proof of Lemma~\ref{lem:cond_low_prob}}
\label{app:proof_low_prob}
   \begin{proof}
       Any arbitrary item from class $c$ can replicate at least \BaseClass{c-1} times only if the number of additional replications of item $i$, due to stochastic replications, exceeds a certain lower bound:
    \begin{equation*}
       \mathcal{D}(c) - \mathcal{D}(c-1) = \frac{\log \pz}{\log \pg}2^c \leq \mathcal{X}(\pg). 
    \end{equation*}
    The probability of this event is less than $\pg^{\frac{\log( \pz)}{\log(\pg)} {2^{c}}}$. Moreover, the maximum number of items in class ${c}$ is $\pz^{-2^{{c}}}$, given that the predicted frequency of items in class ${c}$ is at least $\pz^{2^{c}}$. Consequently, the expected number of items in class ${c}$ that may potentially replicate at least \BaseClass{c-1} times, denoted as $\Exp[\text{Vio}_{c}]$, is bounded by:
    \begin{subequations}
    \begin{align*}
       \Exp[\text{Vio}_c] \leq \pg^{{\log( \pz)}2^c/{\log(\pg)}{}} \cdot \pz^{-2^{c}} \leq 1.
    \end{align*}
\end{subequations}
   \end{proof}

\subsection{Proof of Lemma~\ref{lem:shahinMain}}
\label{app:proof_shahinMain}
\begin{proof}
    We partition the linked lists in \DSName into $K$ \emph{layers}, one layer for each class $c$ as follows. 
 The \emph{layer} of class $0$ is formed by the lists with depth at most $\mathcal{D}(0)$, and for any $c\in [1,K]$, the layer of class $c$ is formed by lists at depth in the range $(\mathcal{D}(c-1),\mathcal{D}(c)]$. Since every item of class $c$ is replicated in at least $\mathcal{D}(K)  - \mathcal{D}(c)$ lists, it appears in layers of all classes $\geq c+1$. 

To bound the number of comparisons for accessing an item $i$, we devise an upper bound for the number of comparisons at any layer $c\leq c_i$.
 Comparisons at layer $c$ involve items that are replicated at $\mathcal{D}(c)$ but not at $\mathcal{D}(c-1)$. This is because items replicated at $\mathcal{D}(c-1)$ have been already examined in previous layers, and upon reaching layer $c$, the search domain is restricted to 2 consecutive items among them. Therefore, all comparisons at layer $c$ involve items in class $c$ as well as items from classes $\geq c+1$ that are replicated at $\mathcal{D}(c)$. By Lemma~\ref{lem:cond_low_prob}, the number of these latter items is at most $1$, on expectation. To conclude, at layer $c$, we have at most $N_c + 1$ items with a replica at level $\mathcal{D}(c)$, each having further replicas following a geometric distribution with parameter $p$. In other words, they resemble a skip list, and searching among them takes at most $\frac{\log (N_c+1)}{p \log{(1/p)}} < \frac{2 \log (N_c)}{p \log{(1/p)}} $, as in a regular skip list~\cite{Pugh90}.

Therefore, for the total search cost, we can write
\begin{align}
\label{eq:exp_search_cost}
 \Exp[\texttt{Search}_{\DSName}(i)]  &=  \frac{2}{p \log{(1/p)}}  \sum_{c \in [0,c_i]}{\log (N_{c})} 
 \\ & = \frac{2}{p \log{(1/p)}} \sum_{[0,c_i]}(-2^{c} \log \theta)  
 \\ & < \frac{2\log \theta}{p \log{p}} 2^{c_i+1}.   
\end{align}
When $c \leq K-1$, the right-hand side of the above inequality is at most $\frac{4}{p \log{p}} \log(\hat{f}_i)$, which follows directly from item classification.

When $c=K$, the total number of comparisons is the sum of comparisons in the layers above $K$ and the number of comparisons at layer $K$. The first term is at most $$ \frac{2\log \theta}{p \log{p}} \cdot  2^{K} \leq \frac{2\log \theta}{p \log{p}} \cdot  2(\log n + \log \theta / \log p); $$ the first inequality follow from Equation~\eqref{eq:exp_search_cost}, applying $c_i = K-1$, and the second inequality follows from the definition of $K$.  Given that there are at most $n$ items in class $K$, the number of comparisons at layer $K$ is at most $\frac{\log n}{- p \log p}$ (as in a regular skip list). Therefore, the total number of comparisons is at most 

\begin{align*}
   \leq&  \frac{2\log \theta}{p \log{p}} 2(\log n + \log \theta / \log p) + \frac{\log n}{-p \log p} \\
   =& \frac{4\log \theta-1}{p \log p} (\log n) + \mathcal{O}(1).
\end{align*} 
\end{proof}

\subsection{Proof of Theorem~\ref{thm:rsl_consistency}}
\label{app:proof_rsl_consistency}
\begin{proof}
    Suppose $\bm{f} = \bm{\hat{f}}$. According to Lemma~\ref{lem:shahinMain}, number of comparison made while searching an item $i$ in class $K$ is at most 
$$\frac{4\log \theta -1}{p \log p} (\log n) + \mathcal{O}(1) \leq \frac{4\log f_i}{p \log p};$$ 
the second inequality from Equation~\eqref{eq:classification}. In addition, again by Lemma~\ref{lem:shahinMain}, the number of comparisons made while searching an item $i$ of class $c_i \leq K-1$ is bounded by $\frac{4\log f_i}{p \log p}$. Therefore, the total number of comparisons for the input sequence is at most $\sum_i f_i \cdot \frac{4\log f_i}{p \log p}$, which ensures a consistency of at most $\frac{-4}{\pg\log(\pg)}$.
\end{proof}

\subsection{Proof of Theorem~\ref{thm:rsl_robustness}}
\label{app:proof_rsl_robustness}
\begin{proof}
    According to Lemma~\ref{lem:shahinMain} the maximum number of comparison made while searching any item is
    \begin{equation*}
        \max_{i} \text{Search}_{\DSName}(i) = \max \bigg(\frac{4}{p \log p} \log(\hat{f_i}), \frac{4\log \theta - 1}{p \log p}(\log n) \bigg) = \Oh(\log n).
    \end{equation*}
    Which shows that \DSName provides robustness of $\Oh(\log n)$, optimal robustness.
\end{proof}

\section{Additional Details of the Experiments}
\label{app:exp}
In this section, we provide additional details of the experimental setup.

Error models play a critical role in the performance of predicted-based data structures such as \DSName or learning-augmented treaps. In \cite{LinLWood22}, the authors define an error as the uncertainty in the frequency of item $i$. Specifically, the noisy predicted frequency of item $i$, $\hat{f_i}$, lies in the range $1/\Delta f_i \leq \hat{f_i} \leq \Delta f_i$, where $\Delta$ is a constant and representative of the error value. However, this error model suffers from the non-linear impact on the performance of predicted-based data structures. Precisely, for the mid-range error values, if $f_i \leq f_j$ denotes the frequency of item $i$ and item $j$, then the expected rank of item $i$ in the predicted ranks is also less than the expected rank of item $j$. As a result, even predictions with mid-range error values slightly lower the performance of data structures like learning-augmented treaps that consider the frequency rankings instead of individual frequency values of items.

To address this challenge, we consider an error model that swaps the ranks of high and low-ranked items for predictions with very high error values. Let $\delta_e$ be a metric that measures the accuracy of frequency predictions, where $\delta_e=0$ indicates perfectly predicted frequencies and $\delta_e=1$ denotes fully adversarial predicted frequencies. Using this error model, the predicted rank of item $i$ is given by $\hat{r_i} = r_i \times (1-\delta_e) + \delta_e \times (n - r_i + 1)$. In other words, the fully adversarial prediction model mirrors the rankings with respect to the median item. During our experiments, we use $\delta_e = 0.01$ as the high-accuracy prediction and $\delta_e = 0.9$ as the low-accuracy prediction.

Existing data structures often exhibit sensitivity to the order of queries, impacting the efficiency of search operations and other related tasks. For instance, splay trees demonstrate varying performance based on the order in which search queries are executed, while the structure of learned treaps is highly influenced by the frequency ordering of elements. It is our expectation that learned treaps will perform well in instances with perfect predictions due to their static optimality, while balanced BSTs will showcase robustness against fully adversarial inputs.

Consequently, we propose an error model that generates fully adversarial input sequences for prediction-based data structures like learned treaps and optimistic optimal binary search trees. These fully adversarial input sequences arise when high-frequency items are swapped with low-frequency items in the prediction result. By incorporating such challenging inputs into our experiments, we aim to assess the optimal robustness and consistency of \DSName under diverse conditions.

\textbf{Dynamic operations during experiments.} Experiments on the synthetic dataset involve insert and delete queries. $80\%$ of the items are randomly selected to be included in the initial data structure, and the ratio of insert/delete queries to search queries is $20\%$. Furthermore, insert and delete queries are randomly interleaved with search queries. The experiments are conducted using two methods: random order with perfect predictions ($\delta_e =0$) and adversarial ranking with noisy predictions.


\subsection{Result of experiments under static setting}
In the static setting of Figure~\ref{fig:staticExp}, the frequency ordering of elements is a crucial factor that can significantly impact the performance of data structures, such as learning-augmented treaps, which take it into account. When the items are ordered according to the adversarial frequencies, the performance of learning-augmented treap is severely affected, especially for a large number of unique keys, even though only $1\%$ of tests included low-accuracy predictions. Results show that testing the dictionary of size $n = 2000$, the average number of comparisons of \DSName is $8.7\%$ more than of learning-augmented treaps under random frequency ordering while \DSName achieves $673.7\times$ lower comparison per query when the predicted frequencies were adversarial. Finally, the result of testing the impact of the Zipfian parameter on the performance of data structures in the static setting is consistent with the result of the same analysis under dynamic conditions.

\begin{figure*}[t]
	\centering
    \subfigure{\includegraphics[width=0.98\textwidth]{figures/Header.png}}\vspace{-4mm}
    \subfigure{\includegraphics[width=0.32\textwidth, height=3.3cm]{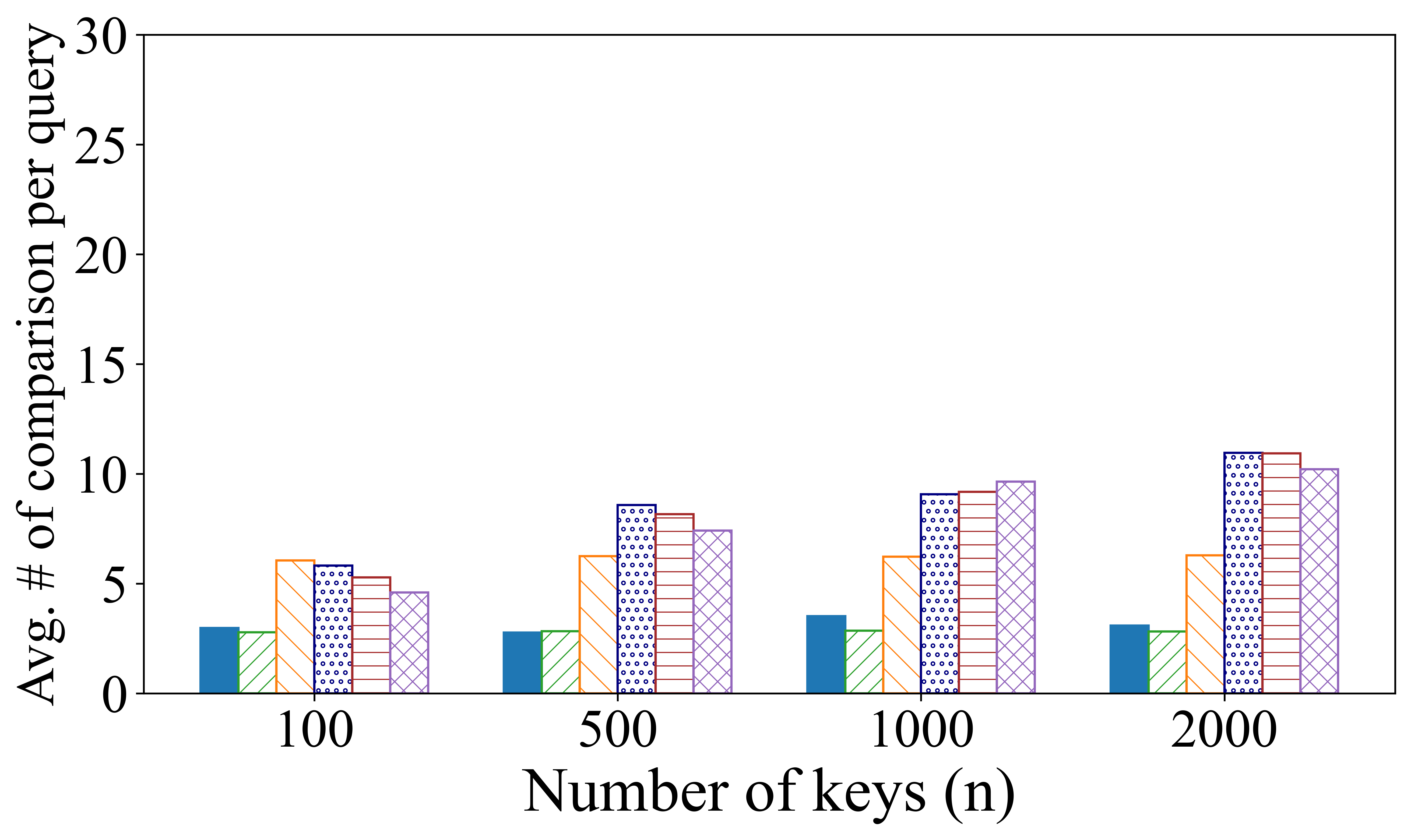}}\hspace{2mm}
	\subfigure{\includegraphics[width=0.32\textwidth, height=3.3cm]{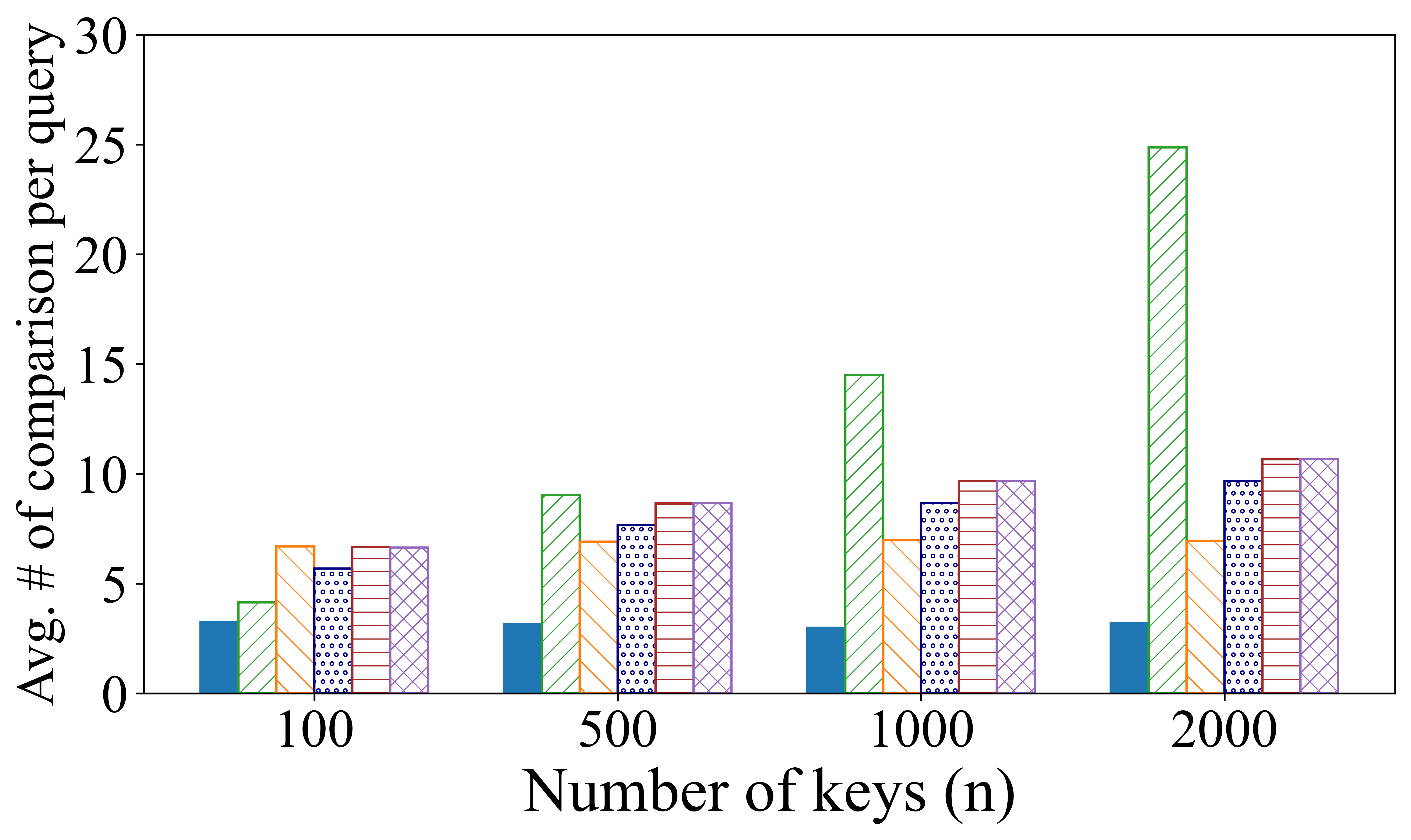}}\hspace{2mm}
    \subfigure{\includegraphics[width=0.32\textwidth, height=3.3cm]{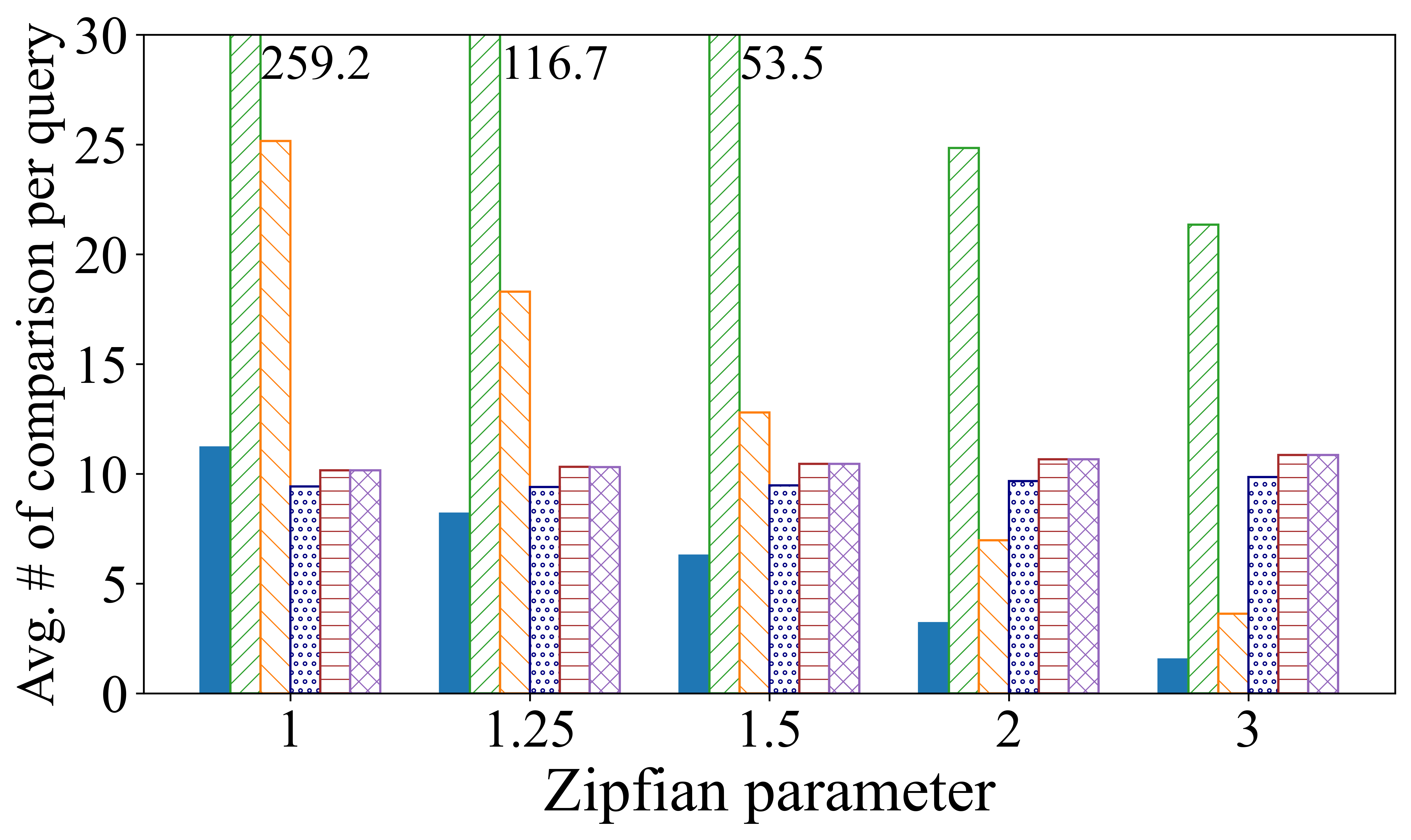}}\vspace{-2mm}
	\caption{Average number of comparisons per query for \DSName and baseline data structures for static evaluations under (left) random frequency ordering with perfect predictions, (center) adversarial frequency ordering with noisy predictions, and (right) adversarial frequency ordering with different value of Zipfian parameter. The frequency ordering of items impacts the performance of learning-augmented treaps while \DSName shows its robustness against different conditions. To improve visibility, the y-axis range in the right figure is limited to 30, with actual values displayed next to bars exceeding this threshold for clarity.}
	\label{fig:staticExp}
    \vspace{-3mm}
\end{figure*}


\subsection{Testing the impact of $\pz$ on performance of \DSName}
\label{app:exp_pz}
Figure~\ref{fig:pz_impact} shows the impact of $\pz$ on the performance of \DSName against the real world dataset (BBC news article dataset). Similar to experiments conducted in Section~\ref{sec:exp_setup}, $0.368$ was selected for parameter $\pg$, $40\%$ of the dataset used as a training dataset, and the size of the adversarial dataset were $25\%$ of the size of the training dataset. The results shows that value of parameter $\pz$ cannot significantly affect the performance of \DSName.

\begin{figure*}[!t]
	\centering
    \includegraphics[width=0.4\textwidth]{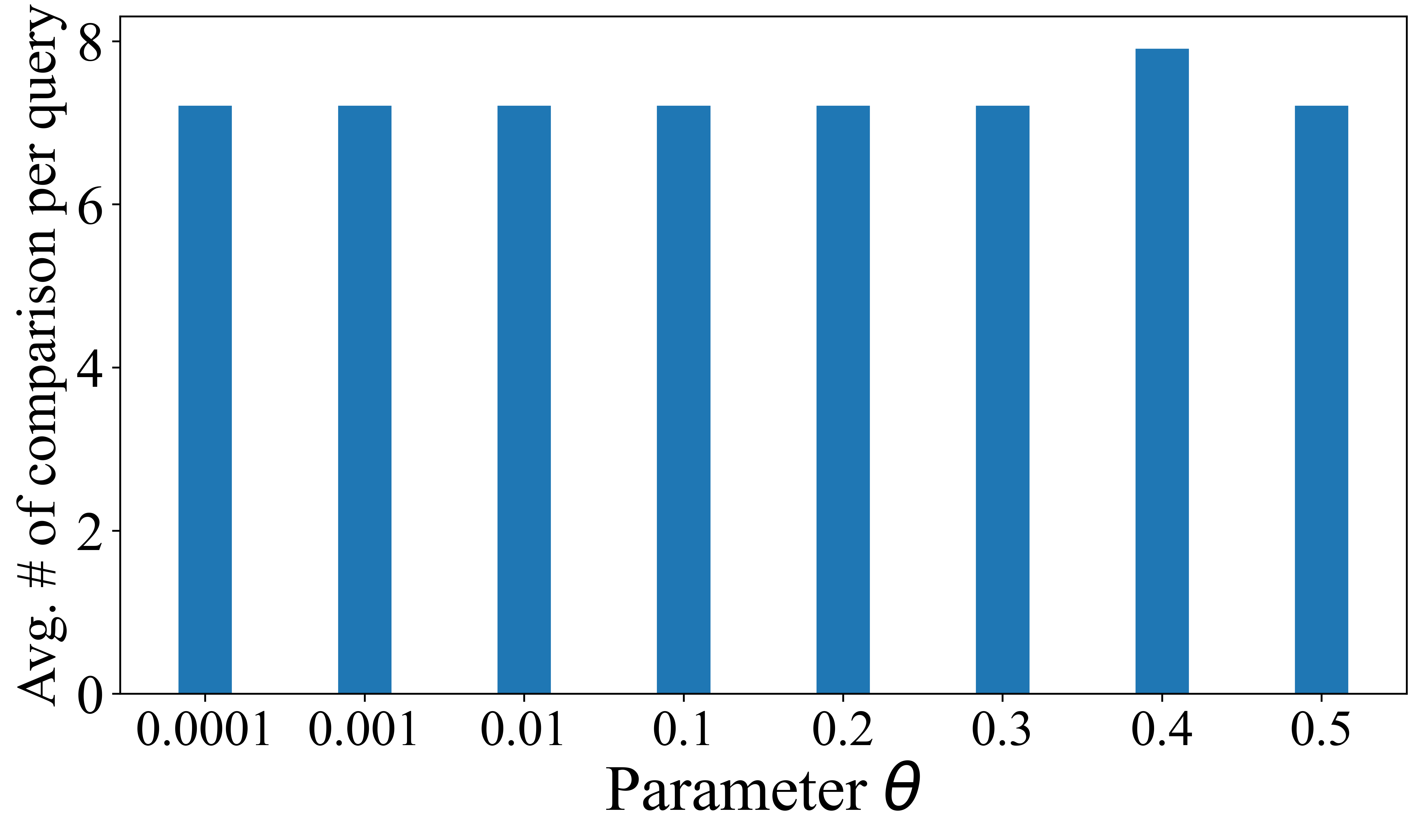}
	\caption{Average number of comparisons per query applied by \DSName as a function of $\pz$ against real world dataset. $\pz$ shows negligible impact on the performance of \DSName.}
	\label{fig:pz_impact}
    \vspace{-3mm}
\end{figure*}

\end{document}